\documentclass{llncs}
\makeatletter
\let\proof\@undefined
\let\endproof\@undefined
\makeatother

\usepackage{enumerate}
\usepackage{url}
\usepackage{ulem}
\usepackage{subfigure}
\usepackage{a4wide}
\usepackage{color}
\usepackage{marginnote}
\usepackage{algorithmic}

\usepackage{amsfonts}
\usepackage{amsthm}

\usepackage[pdftex]{graphicx}
\usepackage{amsmath}

\usepackage{color}

\usepackage{algorithm}

\newcommand{\cX}{{\mathcal X}}

\newcommand{\cC}{{\mathcal C}}
\newcommand{\cT}{{\mathcal T}}

\newcommand{\contains}{{\rightarrow_{\cC}}}

\newcommand{\notcontains}{{\not \rightarrow}_{\cC}}

\newtheorem{observation}{Observation}

\begin{document}
\title{A simple fixed parameter tractable algorithm for computing the hybridization number
of two (not necessarily binary) trees}
\author{Teresa Piovesan, Steven Kelk}
\institute{Department of Knowledge Engineering (DKE), Maastricht University,\\ 
P.O. Box 616, 6200 MD Maastricht, The Netherlands\footnote{Corresponding author is Steven Kelk, steven.kelk@maastrichtuniversity.nl.}
}
\maketitle

\begin{abstract}
Here we present a new fixed parameter tractable algorithm to compute the
hybridization number $r$ of two rooted, not necessarily binary phylogenetic trees on taxon set
$\cX$ in time $(6^r  r!) \cdot poly(n)$, where $n=|\cX|$. The novelty of this
approach is its use of \emph{terminals}, which are maximal elements of a natural
partial order on $\cX$, and several insights from the softwired clusters literature. This yields a surprisingly simple and practical bounded-search algorithm and offers an alternative perspective on the
underlying combinatorial structure of the hybridization number problem.
\end{abstract}

\section{Introduction}

The rooted phylogenetic tree (henceforth, tree) is the traditional model for modelling the evolution of a set of species (or, more generally, \emph{taxa}) $\cX$ (see e.g. \cite{MathEvPhyl,reconstructingevolution,SempleSteel2003}). A rooted phylogenetic
network (henceforth, network) is a generalisation from trees to directed acyclic graphs which allows reticulate evolutionary phenomena such as hybridization, recombination and horizontal gene transfer to be incorporated (see Figure \ref{fig:network}). For detailed background information on networks we refer the reader to \cite{husonetalgalled2009,HusonRuppScornavacca10,surveycombinatorial2011,twotrees,Nakhleh2009ProbSolv,Semple2007}.

\begin{figure}[h]
\centering
\includegraphics[scale=0.2]{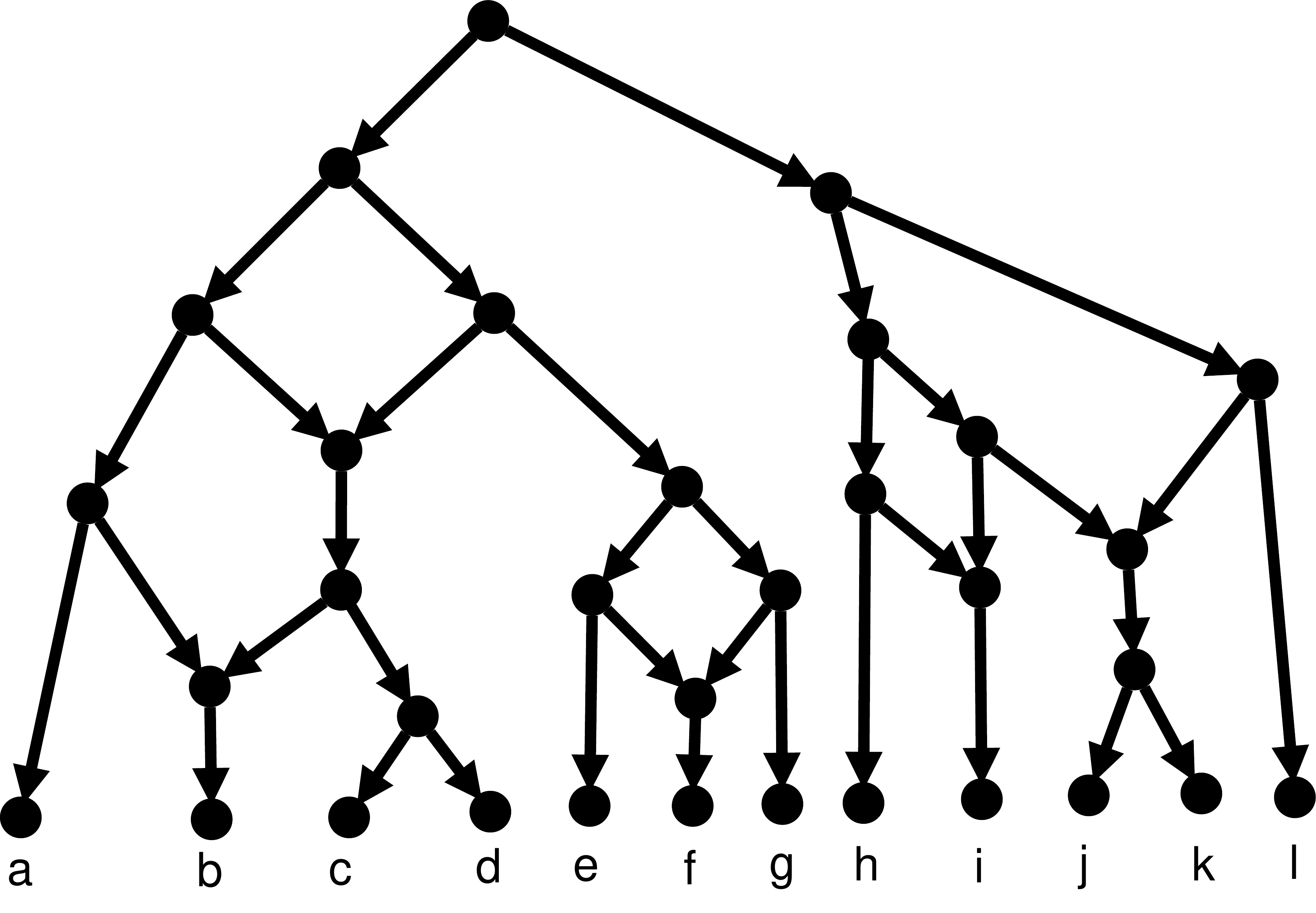}
\caption{An example of a (binary) rooted phylogenetic network on $\cX = \{a, \ldots, l\}$. This network has five reticulation nodes.}
\label{fig:network} 
\end{figure}

One use of networks, motivated in particular by the need to merge a set of discordant gene trees into
a species network \cite{Nakhleh2009ProbSolv}, is the following. Given a set of trees $\cT$, where each tree $T \in \cT$ has the same set of taxa $\cX$, construct a ``most parsimonious'' network which displays all the trees in $\cT$. If we define ``most parsimonious'' to mean: has as few reticulation nodes (i.e. nodes with indegree two or higher) as
possible, we obtain the \emph{hybridization number} problem \cite{baroni05,BaroniEtAl2004}. There has been extensive research into perhaps the simplest possible variant of this problem; this
is when
$\cT$ contains two binary (i.e. fully resolved) trees. Unfortunately, even this stylized version of the problem is computationally difficult; it is NP-hard and in a theoretical
sense difficult to approximate well \cite{bordewich07a,approximationHN}. On the other hand, there has been considerable progress in developing fixed parameter tractable (FPT) algorithms for the problem. Essentially, these are algorithms which can determine whether the hybridization number of two trees is at most $r$ in time $f(r) \cdot poly(n)$, where $n = |\cX|$ and $f(r)$ is a function
that does not depend on $n$ (see \cite{Flum2006} for an introduction to fixed parameter tractability). The idea of such algorithms is that, by decoupling $n$ and
$r$, the running time of the algorithm tends to grow more slowly than algorithms with
a running time of the form $O( n^{f(r)} )$. The first such algorithms were described in \cite{bordewich2,quantifyingreticulation} and the current theoretical state-of-the art is an algorithm with running time $(3.18^{r}) \cdot poly(n)$ \cite{whiddenFixed}. There are also
a number of very fast software packages in existence that are wholly or partially based on insights from fixed parameter tractability \cite{fastcomputation,hybridnet}.

However, what if $\cT$ contains more than two trees and/or contains trees that are not fully resolved? Algorithms to compute the hybridization number of such $\cT$
are necessary, because this more accurately reflects the type of trees that emerge in applied phylogenetics \cite{davidbook}. In this article we are interested in the situation when $\cT$ contains two not necessarily fully resolved trees on $\cX$. (We henceforth refer to such trees as \emph{nonbinary}, noting that this classification includes binary trees as a special case). Given that this problem is a generalisation of the binary case, it inherits all the negative results from that case, but not necessarily the positive results. Indeed,
there are far fewer \emph{positive} results for nonbinary. A number of non-trivial technicalities arise because in the nonbinary case we only require that the network displays \emph{some refinement} of each tree i.e. the image of the tree contained in the network can be more resolved than the original tree \cite{linzsemple2009}. This is a natural and desirable definition given that biologists often use nodes with outdegree 3 or higher in trees to denote \emph{uncertainty}, rather than a hard topological constraint. 

Recently there have been two non-FPT algorithms implemented (both of which are available in the package \textsc{Dendroscope} \cite{Dendroscope3}) to solve the nonbinary problem in polynomial time when the hybridization number is bounded \cite{autumn,cass}. The nonbinary problem is, furthermore, FPT. This was established in \cite{linzsemple2009} using kernelization. Unfortunately, mainly due to the very idiosyncratic behaviour of \emph{common chains} in the nonbinary case, the analysis given in \cite{linzsemple2009} is rather long and complex, and the (weighted) kernel they describe is also rather large, containing at most $(89r)$ taxa; the size of the unweighted kernel is quadratic in $r$. As far as we are aware the algorithm in \cite{linzsemple2009} has not been implemented.

In this article we present an alternative FPT algorithm for nonbinary that is based on bounded-search rather than kernelization, with running time $(6^r r!) \cdot poly(n)$. The resulting algorithm is extremely simple and amenable to implementation (it manages to completely avoid the concept of chains) and the analysis of correctness is comparatively straightforward. The algorithm builds heavily on a number of basic results from the \emph{softwired cluster} literature \cite{husonetalgalled2009,HusonRuppScornavacca10}, in particular \cite{journalElusive}. This literature concerns a slightly different methodology for constructing phylogenetic networks, but as observed in \cite{twotrees,journalElusive} the optima of the models synchronise in the case of two input trees, allowing results and concepts from one methodology to be used in the other.

The simplicity of our new algorithm stems from a careful examination of a natural partial order (and its maximal elements, which we call \emph{terminals}) on $\cX$, which turns out to be closely linked to hybridization number. This partial order appeared earlier in \cite{fptclusters} and \cite{journalElusive} but was used in a slightly different way. Via the observations in \cite{journalElusive} the earlier (and more general) results in \cite{fptclusters} also imply an FPT algorithm via softwired clusters for the nonbinary case, but with an astronomical running time. The added value of the present article is that, by making heavy use of the fact that there only two trees in the input, we are able to obtain a significantly simplified and optimized algorithm that can actually be used in practice.

For completeness we have implemented a prototype version of the algorithm, available upon request. However, perhaps the best use of
the algorithm is to integrate it into existing, well-supported non-FPT algorithms for the nonbinary problem (such as the 2012 release of \textsc{Cass} \cite{blogCass,cass}) to bound their search space and to thus upgrade their status to FPT.


\section{Preliminaries}
\label{sec:prelim}
\subsection{Trees, networks and clusters}

\begin{figure}[h]
\centering
\includegraphics[scale=0.2]{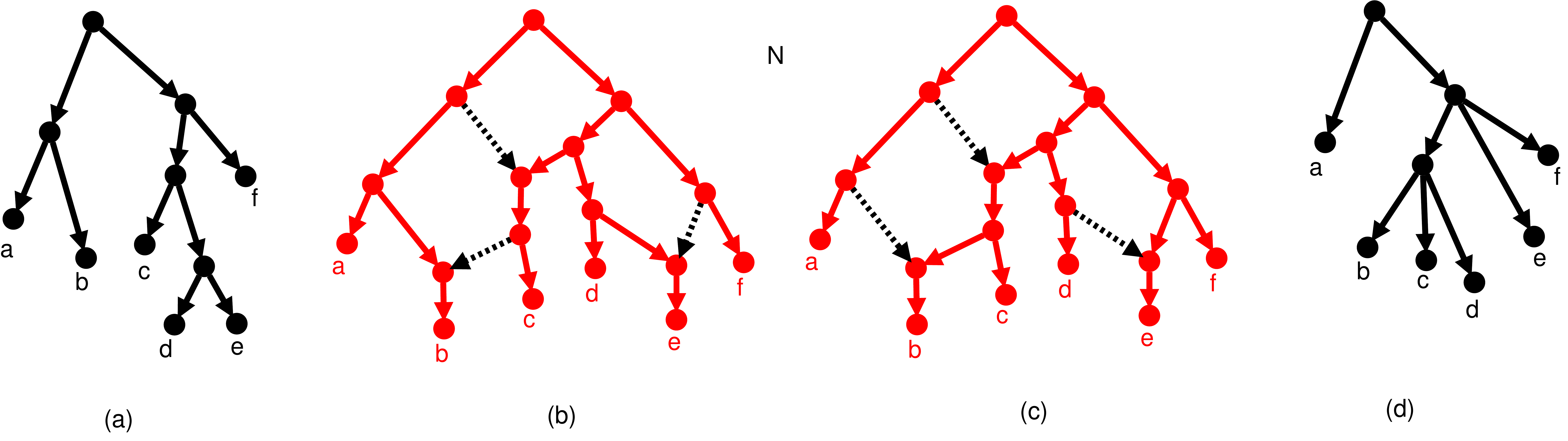}
\caption{In (b) we see that $N$ displays the tree in (a), and in (c) we see that $N$ displays a binary refinement of the tree in (d). The
dotted edges denote the reticulation edges that should be deleted to obtain the required tree.}
\label{fig:display} 
\end{figure}

Consider a set~$\mathcal{X}$ of taxa. A \emph{rooted phylogenetic network} (on~$\mathcal{X}$), henceforth \emph{network}, is a directed acyclic graph with a single node with indegree zero 
(the \emph{root}), no nodes with both indegree and outdegree equal to 1, and leaves bijectively labeled by~$\mathcal{X}$. The indegree of a 
node~$v$ is denoted~$\delta^-(v)$ and~$v$ is called a \emph{reticulation} if~$\delta^-(v)\geq 2$, otherwise it is called a \emph{tree} node. An edge~$(u,v)$ is called a 
\emph{reticulation edge} if its target node
$v$ is a reticulation. When counting reticulations in a network, we count reticulations with more than two incoming 
edges more than once because, biologically, these reticulations represent several reticulate evolutionary events. Therefore, we formally define the \emph{reticulation number} of a 
network~$N=(V,E)$ as \[r(N) = \sum_{\substack{v\in V: \delta^-(v)>0}}(\delta^-(v)-1) = |E| - |V| + 1 \enspace.\]

A \emph{rooted phylogenetic tree} on $\mathcal{X}$, henceforth \emph{tree}, is simply a network that has reticulation number zero. 
We say that a network $N$ on $\mathcal{X}$ \emph{displays}
a tree $T$ if~$T$ can be obtained from $N$ by performing a series of node  and edge deletions and eventually by suppressing nodes with both indegree and outdegree equal to 1 (see Figure \ref{fig:display}). 
We assume without loss of generality that each reticulation has outdegree at least one. Consequently, each leaf has indegree one. We say that a network is \emph{binary} if
every reticulation node has indegree 2 and outdegree 1 and every tree node that is not a leaf has outdegree 2.

Proper subsets of~$\mathcal{X}$ are called \emph{clusters}, and a cluster $C$ is a \emph{singleton} if $|C|=1$. We say that an edge $(u,v)$ of 
a tree \emph{represents} a cluster $C \subset \cX$ if $C$ is the set of taxa descendants of $v$. A tree $T$ represents a cluster $C$ if it contains an 
edge that represents $C$. For example, the tree in Figure \ref{fig:display}(a) represents $\{c,d,e\}$ but not $\{d,e,f\}$. We say that  $N$ represents $C$ ``in the softwired sense'' if $N$ displays some tree $T$ on $\cX$ such that $T$ 
represents $C$. In this article we only consider the softwired notion of cluster representation and henceforth assume this implicitly.  A network represents a set of clusters $\mathcal{C}$ if it 
represents every cluster in $\mathcal{C}$ (and possibly more). 
For a set $\mathcal{C}$ of clusters on $\mathcal{X}$ we define $r(\mathcal{C})$ as $\min \{ r(N) | N \text{ represents } \mathcal{C} \}$, we refer to this as the \emph{reticulation number} 
of $\mathcal{C}$. 
We say that two clusters~$C_1,C_2\subset\mathcal{X}$ are \emph{compatible} 
if either~$C_1\cap C_2=\emptyset$ or~$C_1\subseteq C_2$ or~$C_2\subseteq C_1$, and \emph{incompatible}
otherwise. A set of clusters $\cC$ is compatible if all clusters in $\cC$ are mutually compatible.

\subsection{The equivalence of (maximal) common pendant subtrees and (maximal) ST-sets}

Let $T$ be a tree on $\cX$. We write $Cl(T)$ to denote the
set of clusters represented by edges of $T$, and for a set of trees $\cT$ on $\cX$ we write $Cl(\cT) = \cup_{T \in \cT} Cl(T)$. 
We say that a (binary) tree $T'$ on $\cX$ is a (binary) \emph{refinement} of $T$ if $Cl(T) \subseteq Cl(T')$ (see Figure \ref{fig:display}). We say two trees $T_1$ and $T_2$ on $\cX$ have
a \emph{common refinement} if there exists a tree $T'$ on $\cX$ such that $Cl(T_1) \cup Cl(T_2) \subseteq Cl(T')$, where the last condition is equivalent to saying
that the set of clusters $Cl(T_1) \cup Cl(T_2)$ is compatible. We say that a tree $T^{*}$ on  $\cX^{*} \subseteq \cX$ is a \emph{pendant subtree} of $T$ if there is a refinement $T'$ of $T$ such that $\cX^{*} \in Cl(T')$. Note that this definition does not depend on the topology
of $T^{*}$ so we can equivalently say that $\cX^{*}$ is a pendant subtree of $T$. A pendant subtree $\cX^{*}$ is \emph{non-trivial} if $|\cX^{*}| > 1$. Given two trees $T_1, T_2$ on $\cX$ we say that $\cX^{*} \subseteq \cX$
is a \emph{common} pendant subtree if $\cX^{*}$ is a pendant subtree of both $T_1$ and $T_2$ and $T_1|\cX^{*}$ and $T_2|\cX^{*}$ have a common refinement. (As usual $T|\cX'$ for $\cX' \subseteq \cX$ refers to the tree obtained by suppressing nodes with indegree and outdegree equal to 1 in the
minimal subtree of $T$ that connects all elements of $\cX'$). Note that our definition of common pendant subtree is consistent with \cite{linzsemple2009}, which we follow.

Given a set~$S\subseteq\mathcal{X}$ of taxa,  
 we use~$\mathcal{C}\setminus S$ to denote the result of removing all elements of~$S$ from each cluster
in~$\mathcal{C}$ and we use~$\mathcal{C}|S$ to denote~${\mathcal{C}\setminus (\mathcal{X}\setminus S)}$ (the restriction of~$\mathcal{C}$ to~$S$). Following \cite{journalElusive}, we say that a set~$S \subseteq \mathcal{X}$ 
is an \emph{ST-set} with respect to $\mathcal{C}$, if~$S$ is compatible with all clusters in $\mathcal{C}$ and any two clusters~$C_1,C_2\in\mathcal{C}|S$ are 
compatible.
 An ST-set~$S$ is 
\emph{maximal} if there is no ST-set~$T$ with~$S \subset T$. The maximal ST-sets are unique, partition $\cX$ and can be computed in polynomial time \cite{journalElusive}.

\begin{figure}[h]
\centering
\includegraphics[scale=0.2]{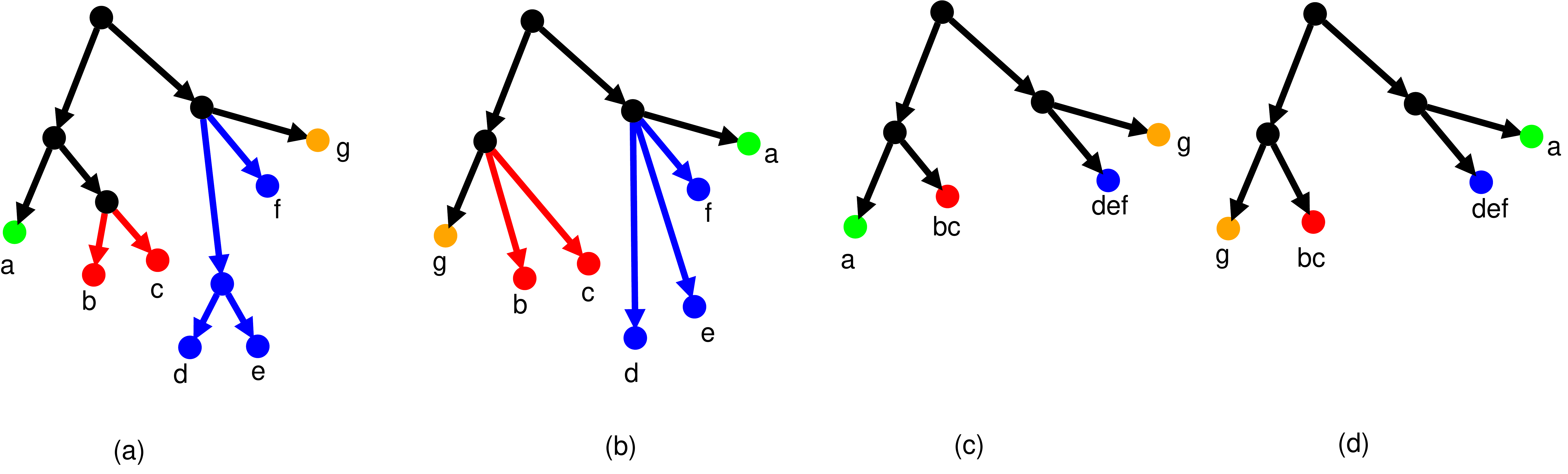}
\caption{The two trees shown in (a) and (b) have maximal ST-sets (i.e. maximal common pendant subtrees) $\{a\}, \{b,c\}, \{d,e,f\}, \{g\}$, shown
in colour. In (c) and (d) we show the result of collapsing the maximal ST-sets in (a) and (b) (respectively) into single taxa.}
\label{fig:subtree} 
\end{figure}

 In \cite[Lemma 6]{journalElusive} it is proven that, if $\cC = Cl(T_1) \cup Cl(T_2)$, and $\cX^{*}$ is an ST-set of $\cC$, then
for each $i \in \{1,2\}$, there exists a node $v_i$ of $T_i$ such that $\cX^{*}$ is exactly equal to the union of the clusters represented by some (not necessarily strict) subset of the edges outgoing from $v_i$. From this it follows that $\cX^{*}$ is a (maximal) ST-set of $Cl(T_1) \cup Cl(T_2)$ if and only if $\cX^{*}$ is a (maximal) common pendant subtree of $T_1$ and $T_2$. We will make heavy use of this equivalence and use the concepts interchangeably. In particular, all the maximal ST-sets 
of $\cC = Cl(T_1) \cup Cl(T_2)$ are singletons (in which case we say $\cC$ is \emph{ST-collapsed} \cite{journalElusive}) if and only if $T_1$ and $T_2$ have no non-trivial common pendant subtrees. A related operation is to create an ST-collapsed set of clusters by \emph{collapsing} all maximal ST-sets into single taxa as shown in Figure
\ref{fig:subtree}. Collapsing maximal ST-sets does not change the reticulation number of the set of clusters (because there always exists an optimal network in which the maximal ST-sets are ``pendant'' \cite[Corollary 11]{journalElusive}).


\subsection{The special case of (clusters obtained from) two trees}

Given two trees $\cT = \{T_1, T_2\}$ on $\cX$, we (again following \cite{linzsemple2009}) define $h(\cT)$ (the \emph{hybridization number} of $\cT$) as the smallest value of $r(N)$ ranging over all networks $N$ on $\cX$ such that $N$
displays a binary refinement of $T_1$ and a binary refinement of $T_2$. In \cite[Observation 9]{journalElusive} we note that the emphasis on \emph{binary} refinements
does not sacrifice generality. Furthermore, from \cite[Lemma 2]{twotrees} we may assume without loss of generality that in the definition of $h(\cT)$, $N$ can be restricted to being binary.
Observe that, if $\cT$ is an arbitrary set of trees on $\cX$, $r(Cl(\cT)) \leq h(\cT)$. This holds because if a network displays a (refinement of a) tree $T$ then it
certainly also represents all the clusters in $Cl(T)$. For $|\cT| > 2$ this inequality can be strict \cite{twotrees}.  
However, in \cite[Lemma 12]{journalElusive} it is proven that
if $\cT = \{T_1, T_2\}$ are two trees on $\cX$, and $\cC = Cl(\cT)$, then $r(\cC) = h(\cT)$. Unfortunately, even in this special case, if $N$ represents all clusters in $\cC$ it does not necessarily display (binary refinements of) $T_1$ and $T_2$ \cite{twotrees}. Fortunately a polynomial-time, reticulation-number preserving transformation
is possible, which we describe later in Section \ref{sec:opt}.

\section{The structure of optimal solutions}
\label{sec:opt}

We begin with some simple results which formalize the idea that, when $\cT$ contains exactly two trees, the problem has ``optimal substructure'' i.e. optimal
solutions can be constructed from arbitrary optimal solutions for well-chosen subproblems. We begin with a focus on clusters, but then explicitly link this to trees in Lemma \ref{lem:unified} and Corollary \ref{cor:clus2tree}.

\begin{observation} 
\label{obs:SBRexists}
Let $\cC = Cl(\cT)$ be a set of clusters on $\cX$, where $\cT = \{T_1,T_2\}$ is a set of two trees on $\cX$ with no
non-trivial common pendant subtrees, and $r(\cC) \geq 1$. Then there exists $x \in \cX$ such that $r( \cC \setminus \{x\} ) < r(\cC)$.
\end{observation}
\begin{proof}
Consider without loss of generality a binary network $N$ which represents $\cC$, where $r(N)=r(\cC)$.
By acyclicity $N$ contains at least one \emph{Subtree Below a Reticulation} (SBR) \cite{journalElusive},
i.e. a node $u$ with indegree-1 whose parent is a reticulation, and such that no reticulation can be reached by a directed
path from $u$.  Let $\cX'$ be the set of taxa reachable from $u$ by directed paths. $\cX'$ is an ST-set, so $|\cX'|=1$ (because
$\cC$ is ST-collapsed). Let $x$ be the single taxon in $\cX'$. Deleting $x$ and its reticulation parent from $N$ (and tidying up
the resulting network in the usual fashion\footnote{Specifically, for as long as necessary
applying the following tidying-up operations until they are no longer needed: deleting any node with outdegree zero that is not labelled by an element of $\cX$; suppressing all nodes with indegree and outdegree both equal to 1;
replacing multi-edges with single edges; deleting nodes with indegree-0 and outdegree-1 \cite{journalElusive}. 
}) creates a network $N'$ on $\cX \setminus \{x\}$ with $r(N') < r(N)$ that represents $\cC \setminus \{x\}$.
\end{proof}

\begin{lemma} 
\label{lem:atmostone}
Let $\cC = Cl(\cT)$ be a set of clusters on $\cX$, where $\cT = \{T_1,T_2\}$ is a set of two trees on $\cX$ with no
non-trivial common pendant subtrees, and $r(\cC) \geq 1$. Then for each $x \in \cX$ it holds that $r(\cC) -1 \leq r( \cC \setminus \{x\} ) \leq r(\cC)$.
\end{lemma}
\begin{proof}
The second $\leq$ is immediate because removing a taxon from a cluster set cannot raise the reticulation number of the cluster set. The first $\leq$ holds
because in \cite[Lemma 10]{journalElusive} it is shown how, given \emph{any} network $N'$ on $\cX \setminus \{x\}$ that represents $\cC \setminus \{x\}$,
we can extend $N'$ to obtain a network $N$ on $\cX$ that represents $\cC$ such that $r(N) \leq r(N')+ 1$.
\end{proof}

\noindent
We recall the following definition from \cite{journalElusive}. For a set of clusters $\cC$ on $\cX$, we call $(S_1, S_2, ..., S_p)$  $(p \geq 0)$
an \emph{ST-set tree sequence of length $p$} if $S_1$ is a ST-set of $\cC$,
$S_2$ is a  ST-set of $\cC \setminus S_1$, $S_3$ is a  ST-set of $\cC \setminus S_1 \setminus S_2$
(and so on) and if all the clusters in $\cC \setminus S_1 \setminus \ldots \setminus S_p$
are mutually compatible i.e. can be represented by a tree. If $\cC = Cl(\cT)$ where $\cT = \{T_1, T_2\}$ are two trees on $\cX$, then
$r(\cC)$ is exactly equal to the minimum length of an ST-set tree sequence for $\cC$ \cite[Corollary 9]{journalElusive}. Essentially, the ST-set tree sequence
describes an order in which common pendant subtrees can be iteratively pruned from $T_1$ and $T_2$ to obtain a common tree $T$. As an example,
the two trees in Figure \ref{fig:subtree}(a) and (b) have a minimum-length ST-set tree sequence $(\{b,c\}, \{d,e,f\})$, and the hybridization number of
these two trees is indeed 2.\\
\\
Observation \ref{obs:SBRexists} and Lemma \ref{lem:atmostone} show that, in an ST-collapsed cluster set, there \emph{always} exists at least one taxon $x$
such that $r(\cC \setminus \{x\}) = r(\cC) - 1$, and that this is the best possible decrease in reticulation number. If we somehow locate such an $x$ (it does \emph{not} matter which one), construct $\cC \setminus \{ x \}$, compute its maximal ST-sets, collapse them, and then repeat this until we obtain a compatible set of clusters,
we are actually constructing a minimum-length ST-set tree sequence $(S_1, \ldots, S_{r(\cC)})$ of $\cC$. (Note that the actual $S_i$ can easily be obtained by reversing any
collapsing operations). Such a sequence not only tells us $r(\cC)$, it also instructs us how to construct in polynomial time a network $N$ which represents
all the clusters in $\cC$ such that $r(\cC) = r(N)$  \cite[Theorem 3]{journalElusive}. Less obviously, it also tells us how to construct a network $N$ with $r(N)=r(C)=h(\cT)$ which displays the two trees that $\cC$ came from:

\begin{lemma}
\label{lem:unified}
Let $\cC = Cl(\cT)$ be a set of clusters on $\cX$, where $\cT = \{T_1,T_2\}$ is a set of two trees on $\cX$. Let $(S_1, \ldots, S_p)$ be an ST-set tree sequence of $\cC$.
Then in polynomial time we can construct a network $N$ that displays binary refinements of
$T_1$ and $T_2$ such that $r(N) = p$.
\end{lemma}
\begin{proof}
(Figure \ref{fig:hang} shows a slightly stylized example of the following).
 Let $\cX_0 = \cX$ and
let $\cX_{i} = \cX_{i-1} \setminus S_{i}$, for $1 \leq i \leq p$. Define $\cC_i = \cC|\cX_{i}$, for
$0 \leq i \leq p$. By assumption,
the clusters in $\cC_{p}$ can be represented by a tree. This is equivalent to
saying that $T_1|\cX_{p}$ and $T_2|\cX_{p}$ have a common refinement.
We construct in polynomial time an arbitrary binary tree $T$ on $\cX_{p}$ that displays these clusters; $T$ will also be a common binary refinement of $T_1|\cX_{p}$ and $T_2|\cX_{p}$. 
Let $T = N_{p}$.
We now show how to construct a network $N_{i-1}$ that displays binary refinements of $T_1|\cX_{i-1}$ and
$T_2|\cX_{i-1}$, given an arbitrary network $N_i$ that displays binary refinements of $T_1|\cX_{i}$ and
$T_2|\cX_{i}$, for $1 \leq i \leq p$. By definition, $S_i$ is an ST-set of $\cC_{i-1}$. $S_i$ thus corresponds to a common pendant subtree of $T_1| \cX_{i-1}$ and
$T_2 | \cX_{i-1}$, and indeed $T_1| \cX_{i}$ and
$T_2 | \cX_{i}$ are exactly the trees obtained by pruning $S_i$ from $T_1| \cX_{i-1}$ and
$T_2 | \cX_{i-1}$. So, reversing this pruning means that $T_1| \cX_{i-1}$ and
$T_2 | \cX_{i-1}$ can be obtained from  $T_1| \cX_{i}$ and
$T_2 | \cX_{i}$ (respectively) by re-grafting $S_i$ at a particular vertex or edge.
Specifically, let $T^{*}$ be an arbitrary binary tree that
represents $\cC_{i-1} | S_i$, this will also be a common binary refinement of the common pendant subtree $S_i$. Now, $N_{i-1}$ can be obtained from $N_{i}$ by extending the
images of $T_1|\cX_{i}$ and $T_2|\cX_{i}$ inside $N_i$ as follows: we introduce
$T^{*}$ below a new reticulation and attach this reticulation at (or, if necessary, slightly above)
the two aforementioned re-grafting points. There are some
small technicalities (such as the need for a ``dummy root'' \cite{journalElusive}) but we omit these details. 
\end{proof}

\begin{figure}[h]
\centering
\includegraphics[scale=0.2]{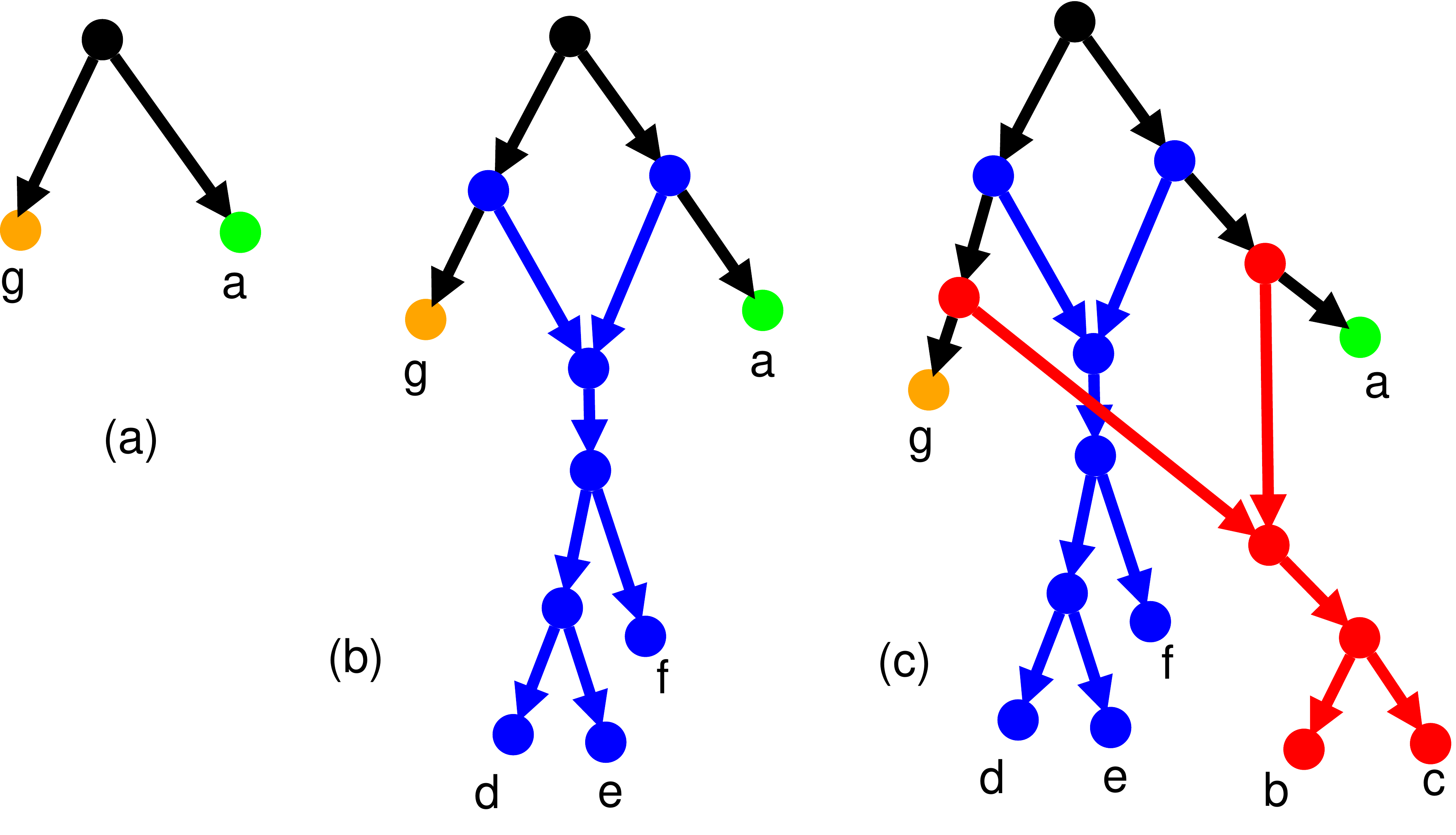}
\caption{A demonstration of the construction described in Lemma \ref{lem:unified}. The trees in Figure \ref{fig:subtree}(a) and (b) have
a minimum-length ST-set tree sequence $(\{b,c\}, \{d,e,f\})$ and here we show how to construct a network $N$ with $r(N)=2$ that
displays binary refinements of both these trees, by re-introducing the elements of the ST-set tree sequence in reverse order.}
\label{fig:hang} 
\end{figure}

\begin{corollary}
\label{cor:clus2tree}
Let $\cC = Cl(\cT)$ be a set of clusters on $\cX$, where $\cT = \{T_1,T_2\}$ is a set of two trees on $\cX$. Let $N$ be a network on $\cX$ that represents all the
clusters in $\cC$. Then in polynomial time we can construct a network $N'$ that displays binary refinements of
$T_1$ and $T_2$ such that $r(N') \leq r(N)$.
\end{corollary}
\begin{proof}
If $N$ is a tree we can simply take a binary refinement of $N$ and we are done. Otherwise, $N$ contains at least one SBR. The taxa in an SBR form an ST-set. So if we identify an SBR of $N$ (which can easily be done in polynomial time), remove it (and tidy up in the usual fashion), and repeat this until we obtain a tree, we obtain an ST-set tree sequence of length at most $r(N)$. (It will be less than $r(N)$ if removing some SBR causes more than one reticulation to disappear from the network when tidying up). This dismantling of $N$ is described in more detail in \cite[Lemma 7]{journalElusive}. We can then apply Lemma \ref{lem:unified} to construct
the network.
\end{proof}

Lemma \ref{lem:unified} and Corollary \ref{cor:clus2tree} allow us for the remainder of the article to focus only on clusters. 

\section{Terminals}

As we have seen, computing $r(\cC)$ (and an accompanying optimal network) essentially
boils down to repeatedly identifying some taxon $x$ such that $r(\cC \setminus \{x\}) = r(\cC) - 1$. The key to attaining fixed parameter tractability is to construct a ``small'' $\cX' \subseteq \cX$ which is guaranteed to contain at least one such taxon $x$. This brings us to the following concept.\\
\\
Given a cluster set $\cC$ and $x,y \in \cX$, we write $x \rightarrow_{\cC}  y$ if and only if every non-singleton cluster in $\cC$ containing $x$, also contains $y$\footnote{Note that, if a taxon $x$ appears
in only one cluster, $\{x\}$, then (vacuously) $x \contains y$ for all $y \neq x$.}. We say that a taxon $x \in \cX$ is a \emph{terminal} if there does not exist $x' \in \cX$ such that $x \neq x'$ and
$x \rightarrow_{\cC} x'$. 

\begin{observation}
\label{obs:poset}
Let $\cC$ be an ST-collapsed set of clusters on $\cX$ such that $r(\cC) \geq 1$. Then the relation $\contains$ is a partial order on $\cX$, the terminals are the maximal
elements of the partial order and each non-singleton cluster of $\cC$ contains at least one terminal.
\end{observation}
\begin{proof}
The relation $\contains$ is clearly reflexive and transitive. To see that it is anti-symmetric, suppose there exist two elements $x \neq y \in \cX$ such that $x \contains y$
and $y \contains x$. Then we have that, for every non-singleton cluster $C \in \cC$, $C \cap \{x,y\}$ is either equal to $\emptyset$ or $\{x,y\}$ i.e. $C$ is compatible
with $\{x,y\}$. Furthermore, the only clusters that can possibly be in $\cC | \{x,y\}$ are $\{x\}, \{y\}$ and $\{x,y\}$ and these are all mutually compatible. So
 $\{x,y\}$ is an ST-set, contradicting the fact that $\cC$ is ST-collapsed. Hence $\contains$ is a partial order. The fact that the terminals are the maximal elements
of the partial order then follows immediately from their definition. Finally, observe
that a non-singleton cluster $C$ must contain at least one terminal, because if it does not then
the relation $\contains$ induces a cycle on some subset of $C$, contradicting the aforementioned anti-symmetry property.
\end{proof}


Let $T$ be a phylogenetic tree on $\cX$. For a vertex $u$ of $T$ we define $\cX(u) \subseteq \cX$ to be the set of all taxa that can be reached from $u$ by directed
paths. For a taxon $x \in \cX$ we define $W^{T}(x)$, the \emph{witness set} for $x$ in $T$, as $\cX(u) \setminus \{x\}$, where $u$ is the parent of $x$. A critical
property of $W^{T}(x)$ is that, for any non-singleton cluster $C \in Cl(T)$ that contains $x$, $W^{T}(x) \subseteq C$ \cite{journalElusive}.

\begin{observation}
\label{obs:termequiv}
Let $\cC = Cl(\cT)$ be a set of clusters on $\cX$, where $\cT = \{T_1,T_2\}$ is a set of two trees on $\cX$ with no
non-trivial common pendant subtrees, and $r(\cC) \geq 1$. Then for any $x \in \cX$ the following statements are equivalent:
(1) $x$ is a terminal of $\cC$;
(2) there exist incompatible clusters $C_1, C_2 \in \cC$ such that $C_1 \cap C_2 = \{x\}$;
(3) $W^{T_1}(x) \cap W^{T_2}(x) = \emptyset$.
\end{observation}
\begin{proof}
We first prove that (2) implies (1). For $x' \not \in C_1 \cup C_2$ it holds that $x \notcontains x'$, because
$x \in C_1$ but $x' \not \in C_ 1$. For $x' \in C_1 \setminus C_2$ it cannot hold that $x \contains x'$, because
$x \in C_2$ but $x' \not \in C_2$, and this holds symmetrically for $x' \in C_2 \setminus C_1$. Hence $x$ is
a terminal. We now show that (1) implies (3).
Suppose (3) does not hold. Then there exists some taxon $x' \in W^{T_1}(x) \cap W^{T_2}(x)$. So
every non-singleton cluster in $\cC$ that contains $x$ also contains $x'$, irrespective of whether the cluster came from $T_1$ or $T_2$.
But then $x \contains x'$, so (1) does not hold. Hence (1) implies (3). Finally, we show that (3) implies (2). Note that (3) implies
that in both $T_1$ and $T_2$ the parent of $x$ is \emph{not} the root. If this was not so, then (wlog) $W^{T_1}(x) = \cX \setminus \{x\}$,
and combining this with the fact that $W^{T_1}(x), W^{T_2}(x) \neq \emptyset$ would contradict (3). 
Hence $W^{T_1}(x) \cup \{x\} \in Cl(T_1)$ and $W^{T_2}(x) \cup \{x\} \in Cl(T_2)$, from which (2) follows.
\end{proof}

For two nodes $u \neq v$ in a network we define a \emph{tree path from $u$ to $v$} as a directed path
that starts at $u$ and ends at $v$ such that all interior nodes of the path are tree nodes. This definition includes the possibility that
$u$ and/or $v$ are reticulation nodes, this will be clear from the specific context. Observe that if $x \neq y$ are taxa in a network $N$ that
represents a set of clusters $\cC$ and there is a tree path from the parent of $x$ to $y$, then $x \contains y$. The set of nodes \emph{reachable by a tree path from $u$} is the set of all $v \neq u$ such that there is a tree path from $u$ to $v$.

\begin{lemma}
\label{lem:3r}
Let $\cC$ be an ST-collapsed set of clusters on $\cX$ such that $r(\cC) \geq 1$. Then 
$\cC$ has at most $3 \cdot  r(\cC)$ terminals.
\end{lemma}
\begin{proof}
Let $N$ be a network on $\cX$ such that $N$ represents $\cC$ and $r(N)=r(\cC)$. Without loss of generality we
can assume $N$ is binary. For each $x \in \cX$, exactly one of the following conditions holds: (1) the parent
of $x$ in $N$ is a reticulation; (2) the parent of $x$ in $N$ is not a reticulation but there is a directed path from the parent of $x$ in $N$ to a reticulation. To see
this observe that if neither condition holds then $N$ contains an edge $(u,v)$ such that at least two taxa, but no reticulations,
are reachable by directed paths from $v$. But then $\cC$ contains a non-singleton ST-set, contradiction. Let $R(N)$ be the
reticulation nodes in $N$. Let $\Omega(\cC) \subseteq \cX$ denote the set of terminals of $\cC$.
We describe a function $F:\Omega(\cC)  \rightarrow R(N)$ such that each reticulation is mapped to at most 3 times, from which the result follows. For each terminal $x$ for which condition (1) holds,
$F(x) = p(x)$, where $p(x)$ is the parent of $x$. For each terminal $x$ for which condition (2) holds, choose a reticulation $r$ such that there is
a tree path from $p(x)$ to $r$, and set $F(x) = r$. Note that there cannot ever be a tree path from $p(x)$ to $y$ if
$x \neq y$ are both terminals, because this would mean $x \contains y$. Now, it follows that a reticulation can be mapped to (in $F$) in at most 3 ways: from a terminal
immediately below it and from one terminal per incoming edge.   
\end{proof}

\begin{corollary}
\label{cor:2rPlus1}
Let $\cC$ be an ST-collapsed set of clusters on $\cX$ such that $r(\cC) \geq 1$. Any subset
of terminals with cardinality $2 \cdot r(\cC) + 1$ or higher, contains at least one taxon $x$ 
such that $r(\cC \setminus \{x\}) < r(\cC)$.
\end{corollary}
\begin{proof}
From the proof of Lemma \ref{lem:3r} we observe that in any subset of $2 \cdot r(\cC) + 1$ terminals,
there exists at least one taxon $x$ for which condition (1) holds. Hence $x$ is an SBR and (as argued in Observation
\ref{obs:SBRexists})  $r(\cC \setminus \{x\}) < r(\cC)$.
\end{proof}

\section{Main result}

\noindent
For a reticulation $r$ in a network $N$, let $\cX^{t}(r)$ be the set of all taxa that can be reached by tree paths from $r$. For example,
if we label the reticulations in the network in Figure \ref{fig:display} $r_1, r_2, r_3$, from left to right, $\cX^t(r_1) = \{b\}, \cX^{t}(r_2) = \{c\}$ and
$\cX^{t}(r_3)=\{e\}$.  The following
lemma shows that an optimal network cannot contain a reticulation $r$ such that $\cX^t(r) = \emptyset$.

\begin{lemma}
\label{lem:emptyNoGood}
Let $\cC = Cl(\cT)$ be a set of clusters on $\cX$, where $\cT = \{T_1,T_2\}$ is a set of two trees on $\cX$ with no
non-trivial common pendant subtrees, and $r(\cC) \geq 1$. Let $N$ be a network on $\cX$ that represents $\cC$ and let $r$ be a reticulation
of $N$ such that $\cX^{t}(r) = \emptyset$. Then $r(\cC) < r(N)$.
\end{lemma}
\begin{proof}
Let $R^t(r)$ be the set of reticulations in $N$ reachable by tree paths from $r$. Now, consider the technique described in the proof of Corollary \ref{cor:clus2tree} for
dismantling $N$ by removing one SBR at a time. All reticulations in $R^t(r)$ will be pruned away at an iteration that is earlier than or equal to the iteration in which
$r$ is pruned away. Moreover, due to the fact that $X^{t}(r) = \emptyset$ - that is, there are no taxa ``sandwiched'' between $r$ and $R^t(r)$ - there definitely exists $r' \in R^t(r)$ such that $r'$ and $r$ both vanish in the same iteration.
But this means that the technique produces an ST-set tree sequence of length strictly less than $r(N)$, which (by Lemma \ref{lem:unified}, or \cite[Theorem 3]{journalElusive}) implies
the existence of a network $N'$ that represents $\cC$ such that $r(N') < r(N)$.  
\end{proof}

\begin{corollary}
\label{cor:onesandwich}
Let $\cC = Cl(\cT)$ be a set of clusters on $\cX$, where $\cT = \{T_1,T_2\}$ is a set of two trees on $\cX$ with no
non-trivial common pendant subtrees, and $r(\cC) \geq 1$. Let $N$ be a network on $\cX$ that represents $\cC$ such that $r(N)=r(\cC)$ and let $r$ be a reticulation
of $N$ such that $\cX^{t}(r) = \{x\}$ for some $x \in \cX$. Then $r(\cC \setminus \{ x \} ) = r(\cC) - 1$.
\end{corollary}
\begin{proof}
If $x$ is an SBR the result is immediate. Otherwise, if $x$ is deleted from $N$, then a network $N'$ is obtained such that $N'$ represents $\cC \setminus \{x\}$ and,
in $N'$, $\cX^{t}(r)=\emptyset$. By Lemma \ref{lem:emptyNoGood} $r(\cC \setminus \{x\}) < r(N')$.  The result follows because $r(N') = r(N) = r(\cC)$.
\end{proof}

For a network $N$, we say that a \emph{switching} of $N$ is obtained by, for each reticulation node, deleting all but one of its
incoming edges. The red subtrees in Figure \ref{fig:display} are switchings. A network $N$ on $\cX$ displays a tree $T$ on $\cX$ if and only if there
is a switching $T_N$ of $N$ such that $T$ can be obtained from $T_N$ by suppressing nodes with indegree and outdegree equal to one (and if necessary
deleting nodes with indegree 0 and outdegree 1). Hence, each switching is the ``image'' in $N$ of some tree displayed by $N$. Indeed, the following
definitions are entirely consistent with the definition of cluster representation given in Section \ref{sec:prelim}. Given a network $N$ and a switching $T_N$ of $N$, we say that an edge $(u,v)$ of $N$ represents a cluster $C$ w.r.t. $T_N$ if $(u,v)$ is an edge of $T_N$ and $C$ is the set of taxa descendants of $v$ in $T_N$. It is natural to define that an edge $(u,v)$ of $N$ represents a cluster $C$ if there exists some switching $T_N$ of $N$ such that $(u,v)$ represents $C$ w.r.t $T_N$.

We say that a cluster $C \in \cC$ is \emph{minimal} if it is a non-singleton cluster such that there does not exist a non-singleton cluster $C' \in \mathcal{C}$ with $C' \subset C$.
\\

\begin{lemma}
\label{lem:goodclus}
Let $\cC = Cl(\cT)$ be a set of clusters on $\cX$, where $\cT = \{T_1,T_2\}$ is a set of two trees on $\cX$ with no
non-trivial common pendant subtrees, and $r(\cC) \geq 1$. There exists a minimal cluster $C \in \cC$ such that, for at least $|C|-1$ of the taxa $x$ in $C$,
$r(\cC \setminus \{ x \} ) = r(\cC) - 1$.
\end{lemma}
\begin{proof}
Let $N$ be a binary network that represents $\cC$ such that $r(\cC) = r(N)$. Let $e=(u,v)$ be an edge of $N$ that represents some non-singleton cluster of $\cC$ such that there does not exist another edge $e'=(u',v')$ reachable from $e$ with this property (where reachable here means: there is a directed path from $v$ to $u'$). Hence
$e$ is a ``lowest'' edge that represents a non-singleton cluster. Let $C \in \cC$ be a non-singleton cluster represented by $e$. We will prove that at least $|C|-1$ taxa $x$ in
$C$ have the property $r(\cC \setminus \{ x \} ) = r(\cC) - 1$. Observe that this property will then automatically also hold for all non-singleton clusters $C' \subset C$, in particular minimal $C'$, from which the claim will follow.

By definition $e = (u,v)$ is an edge of some switching $T_N$ of $N$ such that $C$ is equal to the set of taxa descendants of $v$ in $T_N$. Fix any such $T_N$. Observe
firstly that if there is a directed path in $T_N$ from $v$ to some reticulation $r$, then $\cX^{t}(r) \subseteq C$. The next statement is critical. Suppose there is a tree node $v'$ which is reachable in $T_N$ by a directed path from $v$. Suppose furthermore that, in $T_N$, the set of all taxa $\cX'$ reachable from $v'$ by tree paths (in $T_N$) has cardinality exactly 2. We show that this situation cannot actually happen. To see this, let $\{y,z\}$ be the taxa in $\cX'$. By assumption $\{y, z\}$ is not an ST-set, because
$\cC$ is ST-collapsed. Hence
there must exist a non-singleton cluster $C^{*} \in \cC$ such that without loss of generality $C^{*} \cap \{y,z\} = \{y\}$. Now, $C^{*}$ must be represented
by some edge $e' = (u',v')$ of $N$. Moreover, $e'$ must lie somewhere on the tree path from $v'$ to $y$ in $T_N$. However, $u'$ is then reachable by a directed path
from $v$, contradicting our claim that $e$ was ``lowest''. So such an $\cX'$ does not exist. Now, suppose that $r$ is a reticulation in $T_N$ such that (1) $r$ can
be reached in $T_N$ by a directed path from $v$, (2) two or more taxa can be reached in $T_N$ from $r$ by tree paths. Due to the fact that $N$ is binary, there must
exist a tree node $v'$ reachable in $T_N$ by a tree path from $r$, such that $\{x,y\}$ are the only two taxa reachable from $v'$ by tree paths in $T_N$. We have
already concluded, however, that this is not possible. Hence we can infer that,  if $r$ is a reticulation in $T_N$ such that $r$ can be reached by a directed path
from $v$, $|\cX^t(r)|=1$. This, in turn, means that with one possible exception (because there can be at most one taxon in $C$ reachable in $T_N$ from $v$ by a tree path) each taxon $x \in C$ is such that $\cX^{t}(r) = \{x\}$ for some $r$ i.e. $x$ is either an SBR or is the unique taxon ``sandwiched'' between several reticulations.
By Corollary \ref{cor:onesandwich} we are done.
\end{proof}

An immediate consequence of Lemma \ref{lem:goodclus} is that if we could identify the minimal cluster $C$, it would be sufficient to restrict our attention
to an arbitrary size-2 subset of it: we could still be sure that at least one of the the taxa $x$ is such that $r(\cC \setminus \{ x \} ) = r(\cC) - 1$. This is the motivation behind the following theorem.

\begin{theorem}
\label{thm:bigdog}
Let $\cC = Cl(\cT)$ be a set of clusters on $\cX$, where $\cT = \{T_1,T_2\}$ is a set of two trees on $\cX$ with no
non-trivial common pendant subtrees, and $r(\cC) \geq 1$. Let $\cX' \subseteq \cX$ be the set constructed as follows. If there are strictly more than
$2 \cdot r(\cC)$ terminals in $\cC$, let $\cX'$ be an arbitrary subset of the terminals of cardinality $2 \cdot r(\cC)+1$. Otherwise, for each minimal cluster
$C \in \cC$, put two arbitrary taxa from $C$ in $\cX'$, of which at least one is a terminal. Then $|\cX'| \leq 6 \cdot r(\cC)$ and there
exists $x \in \cX'$ such that $r(\cC \setminus \{ x \} ) = r(\cC) - 1$.
\end{theorem}
\begin{proof}
The first way of constructing $\cX'$ is correct by Corollary \ref{cor:2rPlus1}. Let us then assume that there are at most $2 \cdot r(\cC)$ terminals. Recall
that each (minimal) cluster contains at least one terminal, by Observation \ref{obs:poset}. A terminal can appear in at most one minimal cluster from $T_1$, and
at most one minimal cluster from $T_2$. Consider the following mapping from $\cX'$ to itself. Map each terminal to itself. For each non-terminal $y \in \cX'$, map $y$ (arbitrarily) to a terminal $x \in \cX'$ such that $x$ and $y$ are both in some minimal cluster of $\cC$. In this mapping, a terminal can be mapped onto at most 3 times (i.e. from itself and at most two non-terminals). Hence $|\cX'| \leq 6 \cdot r(\cC)$.
\end{proof}
 
\section{The algorithm}
\label{sec:alg}

We describe the algorithm non-deterministically to keep the exposition as clear as possible.\\ 
\\
 \textbf{Input:} Two trees $\cT = \{T_1, T_2\}$ on the same set of taxa $\mathcal{X}$.\\
 \textbf{Output:} A network $N$ that displays binary refinements of $T_1$ and $T_2$ such that $r(N)=h(\cT)$.
\begin{algorithm}
\caption{}
\begin{algorithmic}[1]
\STATE set $\cC := Cl(\cT)$
\STATE guess $r = h(\cT) = r(\cC)$ 
\FOR {$i := r$ \textbf{downto} $1$}
 \STATE collapse all maximal ST-sets (i.e. maximal common pendant subtrees) in $\cC$ to obtain a set of clusters $\cC'$
\IF {$\cC'$ contains more than $2i$ terminals}
\STATE {set $\mathcal{X'}$ to be an arbitrary size $2i+1$ subset of the terminals}
\ELSE
\STATE {construct $\cX'$ by taking two taxa from each minimal cluster of $\cC'$, such that at least one of each pair is a terminal}
\ENDIF
\STATE { guess an element $x \in \cX'$ such that $r(\cC' \setminus \{x\}) = r(\cC') - 1$ and record that $x_{r-i+1} := x$ }
\STATE { set $\cC := \cC' \setminus \{x\}$}
\ENDFOR
\STATE convert the sequence $(x_1, \ldots, x_r)$ into the ST-set tree sequence $\mathcal{S} = (S_1, \ldots, S_r)$ of $\cC$ by decollapsing taxa
\STATE use $\mathcal{S}$ to construct a binary network $N$ with $r(N) = h(\cT)$ that displays binary refinements of $T_1$ and $T_2$ (see Lemma \ref{lem:unified}).
\end{algorithmic}
\end{algorithm}
\\
The correctness of the algorithm is primarily a consequence of Lemma \ref{lem:goodclus} and Corollary \ref{cor:2rPlus1}. If we let $r = r(\cC)$, the running time is at most $(6^r r!) \cdot r \cdot poly(n)$ where $n = |\cX|$. The single $r$ term comes from line 2. The $(6^r r!)$ term is a consequence of Theorem \ref{thm:bigdog}; $|\cX'|$ never rises above $6r$, and each iteration of the main loop is assumed to reduce the reticulation number by 1, giving a running time of at most
$(6r)(6(r-1))(6(r-2))\ldots = 6^r r!$. The $poly(n)$ term includes operations such as computing terminals, locating minimal clusters and collapsing maximal ST-sets; the first two operations are clearly polynomial-time because $\cC(\cT) \leq 4(n-1)$ (which follows from the fact that a tree on $n$ taxa contains at most $2(n-1)$ edges). In fact, the most time-consuming operation inside the $poly(n)$ term is collapsing maximal ST-sets (i.e. maximal common pendant subtrees).  In \cite[Lemma 5]{journalElusive} a naive $O(n^4)$ algorithm is given for this although with intelligent use of data structures and exploiting the fact that $\cC$ comes from two trees $O(n^2)$ is certainly possible without too much effort. Finally, we note that the single $r$ term can be absorbed, if necessary, into the $poly(n)$ term to give $(6^r r!) \cdot poly(n)$, because (trivially) $r \leq n$.

\section{Acknowledgements}
We gratefully acknowledge Jean Derks and Nela Lekic for their helpful comments concerning an earlier version of this article. We also thank Simone Linz and
Leo van Iersel for useful discussions.

\bibliography{6r2012}

\end{document}